\documentclass[letterpaper]{article} 
\usepackage[preprint]{aaai2027}  
\usepackage[hyphens]{url}  
\usepackage{graphicx} 
\usepackage{natbib}  
\usepackage{caption} 
\usepackage{algorithm}
\usepackage{algorithmic}

\usepackage{booktabs}

\usepackage{amssymb}
\usepackage{mathtools}
\usepackage{amsthm}
\usepackage{amsbsy}
\usepackage{dsfont} 
\newcommand{\mathbbm}[1]{\mathds{#1}} 
\usepackage{multirow}

\theoremstyle{plain}
\newtheorem{theorem}{Theorem}[section]

\theoremstyle{definition}

\newtheorem{assumption}[theorem]{Assumption}
\theoremstyle{remark}

\DeclareMathOperator*{\argmax}{arg\,max}

\newcommand{\E}{\mathbb{E}}

\newcommand{\veps}{\varepsilon}

\def\ddefloop#1{\ifx\ddefloop#1\else\ddef{#1}\expandafter\ddefloop\fi}
\def\ddef#1{\expandafter\def\csname bb#1\endcsname{\ensuremath{\mathbb{#1}}}}
\ddefloop ABCDEFGHIJKLMNOPQRSTUVWXYZ\ddefloop
\def\ddef#1{\expandafter\def\csname b#1\endcsname{\ensuremath{\mathbf{#1}}}}
\ddefloop ABCDEFGHIJKLMNOPQRSTUVWXYZ\ddefloop
\def\ddef#1{\expandafter\def\csname c#1\endcsname{\ensuremath{\mathcal{#1}}}}
\ddefloop ABCDEFGHIJKLMNOPQRSTUVWXYZ\ddefloop

\title{Exponential Reward Weighting for Fine-Tuning Generative Recommenders under Sparse and Noisy Feedback}
\author{
    Keertana Chidambaram\textsuperscript{\rm 1}\thanks{Work done while KC was an intern at Netflix Research.},
    Sanath Kumar Krishnamurthy\textsuperscript{\rm 2}\equalcontrib,
    Qiuling Xu\textsuperscript{\rm 3}\equalcontrib,\\
    Ko-Jen Hsiao\textsuperscript{\rm 3},
    Moumita Bhattacharya\textsuperscript{\rm 3}
}
\affiliations{
    \textsuperscript{\rm 1}Department of Management Science \& Engineering, Stanford University\\
    \textsuperscript{\rm 2}Meta\\
    \textsuperscript{\rm 3}Netflix Research\\
    vck@stanford.edu, sanathkk@meta.com, qxu@netflix.com, kojenh@netflix.com, moumitab@netflix.com
}

\begin{document}

\maketitle

\begin{abstract}
In recommendation systems, users interact with only a small fraction of a vast item catalog, producing feedback that is both sparse and noisy. This challenges post-training generative recommenders: reward models trained from logged interactions often fail to generalize, while directly optimizing imperfect rewards can lead to reward over-optimization. We propose Exponential reward-weighted fine-tuning (Exp-RSFT), where each logged interaction is weighted by $\exp(r/\lambda)$, avoids this failure by optimizing directly on the logged rewards, with the temperature $\lambda$ regularizing against their noise. We theoretically show that Exp-RSFT’s suboptimality decomposes into two costs: a coverage cost arising from limitations of the logging policy and a noise cost from imperfect feedback. The temperature $\lambda$ balances these competing effects, yielding an optimal tradeoff between exploiting high-reward behavior and robustness to noise. Across three public benchmarks and a large-scale industrial dataset, we verify this theoretical prediction: performance follows an inverted-U trend as a function of $\lambda$, while PPO and DPO often over-optimize unreliable reward models and degrade recommendation quality. Exp-RSFT consistently improves ranking performance without requiring online exploration or preference data.
\end{abstract}

\section{Introduction}
Generative recommenders cast recommendation as sequential generation over a user's interaction history, mirroring next-token prediction in LLMs~\cite{yang2025gr, rajput2023recommender, zhai2024actions, deng2025onerec}. Architectures such as SASRec~\cite{kang2018self}, HSTU~\cite{zhai2024actions}, and OneRec~\cite{deng2025onerec} are trained by behavior cloning, which imitates high-value engagements (genuine enjoyment) and low-value ones (accidental clicks, click-bait) indiscriminately.

Reinforcement Learning from Human Feedback (RLHF)~\cite{ouyang2022training, stiennon2020learning} shows that, for LLMs, post-training on real user feedback improves outcomes beyond what behavior cloning alone can achieve. Recommendation is a natural fit for this paradigm because organic feedback is available at scale through ratings, reviews, watch time, and re-engagement, which is precisely the signal needed to distinguish high-value from low-value engagements. However, adapting RLHF to generative recommenders introduces unique challenges:

\begin{itemize}
    \item \textbf{Sparse user interactions.} In recommendation settings with large catalogs, users interact with only a tiny fraction of the available items, and item representations are learned purely from user behavior with no semantic grounding~\cite{dulac2015deep}. A reward model fit on such sparse data must extrapolate across a vast item space, and generalizes poorly. During post-training, the policy exploits the resulting prediction errors, known as \emph{reward hacking} \cite{amodei2016concrete}, systematically selecting items for which the reward model is unreliable rather than those users actually prefer.

    \item \textbf{Noisy rewards.} Even the observed rewards are imperfect measures of preference: ratings for the same item vary across sessions, clicks are affected by position bias, and implicit signals fluctuate for reasons unrelated to true preference. Any method that optimizes directly on logged rewards must therefore guard against overfitting to this noise.

    \item \textbf{No online exploration.} Industrial datasets are pre-collected and static, so the policy cannot query users for new feedback mid-training. Standard RLHF pipelines~\cite{ouyang2022training, stiennon2020learning} sidestep this by training a reward model as a simulator, but in our setting that simulator inherits the generalization failures above, reintroducing the very problem we set out to avoid.

    \item \textbf{Unknown logging policy.} Rewards are observed only for items previously shown to the user by the deployed (logging) policy, inducing selection bias. Inverse Propensity Scoring (IPS) could correct this in principle, by re-weighting examples by their probability of being shown, but the logging policy is typically too complex~\cite{covington2016deep} or inaccessible~\cite{liang2022local} to estimate, and IPS weights have extreme variance~\cite{swaminathan2015batch, dudik2011doubly}, especially when the logging policy is near-deterministic.
\end{itemize}

To overcome these challenges, we adopt \textbf{Exponential Reward-Weighted SFT (Exp-RSFT)} where we weight each logged trajectory $\tau$ by $\exp(R(\tau)/\lambda)$, where $R(\tau)$ is the observed cumulative reward of the trajectory and $\lambda > 0$ is a temperature parameter. Because Exp-RSFT never queries a learned reward model, it cannot reward hack; it needs no propensity scores; and it is fully offline. The algorithm itself is well known~\cite{peters2007using, wang2018exponentially, liang2022local}; our contribution is a theoretical account of when and why it succeeds in this regime, together with experiments confirming the theory's predictions:

\begin{itemize}
    \item \textbf{Theoretical sub-optimality decomposition.} We prove that the policy learned by Exp-RSFT competes with any alternative policy one might hope to imitate, including the optimal one, and that its suboptimality against the optimal policy splits into two interpretable terms: a coverage cost that grows only logarithmically with the inverse probability of the optimal trajectory under the logging policy, and a noise cost that grows only logarithmically with the catalog size. The analysis is trajectory-level, requires no reward model or value function, and recovers the contextual bandit as the special case $T = 1$.

    \item \textbf{Temperature as a trade-off parameter.} The two costs move in opposite directions in $\lambda$, while coverage favors small temperatures, noise robustness favors large ones. Their sum is minimized at a closed-form $\lambda^\star$, so the theory predicts an inverted-U in performance as $\lambda$ varies.

    \item \textbf{Experiments validating theory.} Across three public benchmarks and a large-scale proprietary dataset, temperature sweeps trace the predicted inverted-U, and Exp-RSFT outperforms behavior cloning, linear reward weighting (RSFT), PPO, and DPO; the gains over RSFT isolate the temperature as the operative ingredient. Consistent with the regime our theory targets, the learned reward model fails to outperform naive baselines in our setup, and PPO and DPO collapse by over-optimizing it.
\end{itemize}

\section{Related Literature}

\textbf{Post-training generative recommenders.} Standard post-training methods include RLHF~\cite{ouyang2022training, stiennon2020learning}, GRPO~\cite{shao2024deepseekmath}, and DPO~\cite{rafailov2023direct}. RLHF and GRPO sample outputs from the current policy and score them with a reward model, verifier, or simulator; in recommendation there is no ground-truth correctness to verify and no live simulator of user responses, so the score must come from a learned reward model that extrapolates across a vast catalog from sparse observations, inviting the reward hacking we document in Section~\ref{sec:experiments}. DPO sidesteps reward-model training but requires (preferred, dispreferred) pairs, which scalar feedback such as ratings and watch time does not provide; constructing pairs by ranking candidates with a reward model reintroduces the same problem. Beyond the reward-modeling problem, logged data only contains feedback on items the logging policy chose to show. Classical offline policy learning corrects for this with inverse propensity scores~\cite{swaminathan2015batch, dudik2011doubly} or learns conservative value functions~\cite{kumar2020conservative}, but propensities are unavailable under complex production logging policies, and value functions must again generalize across the catalog. Behavior cloning needs none of these but cannot improve on the logging policy. This leaves methods that optimize directly on logged rewards while staying close to the logging policy as the natural candidates for our regime.

\textbf{Exponential reward weighting.} The algorithm we analyze is well established: RWR~\cite{peters2007using}, MARWIL~\cite{wang2018exponentially}, IQL~\cite{kostrikov2021offline}, CRR~\cite{wang2020critic}, and LPI~\cite{liang2022local} all weight log-likelihoods by monotonic, typically exponential, functions of returns or advantages. IQL, CRR, and MARWIL build advantages from learned value functions; RWR weights by returns directly, as we do, avoiding value-function estimation altogether (Section~\ref{sec:background}). Relative to the linear reward weighting of \citet{mukherjeeoffline}, the exponential form adds the temperature $\lambda$ as an explicit robustness knob and avoids the negative-weight pathology of linear schemes in likelihood-based objectives~\cite{steck2010training, schnabel2016recommendations}. Existing analyses of this family address optimization stability, convergence, and distribution shift in value-function estimation; none characterizes behavior under reward noise or quantifies the cost of limited coverage in large action spaces. Providing exactly this characterization is the contribution of Section~\ref{sec:theory}.

\textbf{Reward weighting in LLM post-training.} RLOO~\cite{ahmadian2024back} samples multiple outputs per prompt and weights by the group-relative reward, and ALoL~\cite{baheti2023leftover} requires propensity scores and a learned value function; neither transfers to our setting, where fresh feedback cannot be sampled and propensities are unknown.

\section{Background}
\label{sec:background}

\subsection{Problem Setting}
\label{sec:setup}
A user's interaction history serves as the initial context $s_1 \sim d_0$. Starting from $s_1$, a policy generates a \emph{trajectory} of $T$ recommendations: at step $t$ it selects an item $a_t \in \mathcal{A}$ given the current state $s_t$, and the state updates deterministically by appending the selected item, $s_{t+1} = (s_t, a_t)$. A policy $\pi$ thus induces a distribution over trajectories $\tau = (s_1, a_1, \ldots, s_T, a_T)$, written equivalently in product or summation form:
\begin{align}
\pi(\tau \mid s_1) &= \prod_{t=1}^{T} \pi(a_t \mid s_t), \nonumber \\
\log \pi(\tau \mid s_1) &= \sum_{t=1}^{T} \log \pi(a_t \mid s_t).
\label{eq:traj_policy}
\end{align}
Each trajectory receives a scalar reward $R(\tau)$ summarizing the user's feedback over the recommendation session. We treat $R$ as an arbitrary functional of the trajectory: it may be computed from feedback observed at each recommendation step, intermediate feedback collected throughout the trajectory, terminal session-level feedback, or any other evaluation metric. Our analysis only requires that the resulting feedback can be represented as a scalar reward associated with each trajectory. We assume access to an offline dataset $\mathcal{D} = \{\tau_i\}_{i=1}^{n}$ of trajectories with observed rewards, collected under an unknown logging (behavior) policy $\pi_\beta$. The value of a policy from context $s$ is
\begin{equation}
V^{\pi}(s) = \mathbb{E}_{\tau \sim \pi(\cdot \mid s)}\!\left[R(\tau)\right].
\end{equation}
The contextual bandit setting is the special case $T = 1$, where $\tau = (s, a)$ and $R(\tau) = r(s, a)$.

\subsection{Exponential Reward-Weighted Policy Optimization}
\label{sec:exp_rw}

Following \citet{nair2020awac}, we seek a policy that maximizes expected reward while remaining close to the data-generating distribution:
\begin{align}\label{eq:constrained_opt}
&\max_{\pi}\; \E_{s \sim d_0,\, \tau \sim \pi(\cdot|s)}\!\left[ R(\tau) \right] \\
&\text{s.t.} \quad \E_{s \sim d_0}\!\left[ D_{\mathrm{KL}}\!\left( \pi(\cdot|s) \,\|\, \pi_\beta(\cdot|s) \right) \right] \leq \veps, \nonumber
\end{align}
where the optimization is over distributions on trajectories. Forming the Lagrangian with multiplier $\lambda$ and setting the functional derivative to zero yields the closed-form solution
\begin{equation}
\pi^*(\tau|s) = \frac{1}{Z(s)}\, \pi_\beta(\tau|s) \exp\!\left( \frac{R(\tau)}{\lambda} \right),
\label{eq:reward_weighted_policy}
\end{equation}
where $Z(s) = \sum_{\tau'} \pi_\beta(\tau' \mid s)\exp\!\left(R(\tau')/\lambda\right)$ is the partition function, with the summation taken over all trajectories $\tau'$ originating from the initial state $s$. Equation~\eqref{eq:reward_weighted_policy} can be viewed as an exponential tilt of the behavior policy toward high-reward trajectories. Since generative recommenders define autoregressive distributions over trajectories, the tilted distribution $\pi^*$ remains within the same autoregressive policy class.
To obtain a parametric policy $\pi_\theta$, we project $\pi^*$ onto the policy class by minimizing $D_{\mathrm{KL}}(\pi^* \| \pi_\theta)$. Rewriting the expectation under $\pi_\beta$ via the density ratio $\pi^*/\pi_\beta$ and dropping the partition function as in \citet{nair2020awac} gives the weighted maximum-likelihood objective
\begin{equation}
\theta^* = \argmax_{\theta}\; \E_{\tau \sim \mathcal{D}} \left[ \exp\!\left( \frac{R(\tau)}{\lambda} \right) \sum_{t=1}^{T} \log \pi_\theta(a_t \mid s_t) \right],
\label{eq:rw_objective}
\end{equation}
where the inner sum is the trajectory log-likelihood from Equation~\eqref{eq:traj_policy}. This is Algorithm~\ref{alg:exp-rw-sft}: standard supervised fine-tuning in which each logged trajectory is weighted by its exponentiated, temperature-scaled reward. It requires no reward model, no propensity scores, and no knowledge of $\pi_\beta$; the single hyperparameter $\lambda$ controls regularization strength. In the next section, we analyze the behavior of exponential reward weighting under noisy reward observations. We derive high-probability improvement guarantees relative to arbitrary comparator policies and characterize how the temperature $\lambda$ balances reward optimization, noise sensitivity, and offline data coverage.

\begin{algorithm}[tb]
\caption{Exponential Reward-Weighted SFT (Exp-RSFT)}
\label{alg:exp-rw-sft}
\begin{algorithmic}
\STATE {\bfseries Input:} Offline dataset $\mathcal{D} = \{\tau_i\}_{i=1}^{n}$ with rewards $R(\tau_i)$, initial policy $\pi_{\theta}$, temperature $\lambda > 0$, epochs $E$
\FOR{$e = 1, \ldots, E$}
    \FOR{each mini-batch $B \subset \mathcal{D}$}
        \STATE $\mathcal{L}(\theta) \gets -\frac{1}{|B|} \displaystyle\sum_{\tau \in B} \exp\!\left(\frac{R(\tau)}{\lambda}\right) \sum_{t=1}^{T} \log \pi_\theta(a_t \mid s_t)$
        \STATE Update: $\theta \gets \theta - \eta \nabla_\theta \mathcal{L}(\theta)$
    \ENDFOR
\ENDFOR
\STATE {\bfseries Output:} Fine-tuned policy $\pi_{\theta}$
\end{algorithmic}
\end{algorithm}

\section{Theoretical Analysis}
\label{sec:theory}
We analyze exponential reward weighting from three perspectives: improvement over the behavior policy, robustness to noisy user feedback, and the trade-off between reward optimization and offline data coverage. Throughout, we fix a context $s$ and let $\mathcal{T}$ denote the set of trajectories reachable from $s$, with $|\mathcal{T}| \leq |\mathcal{A}|^T$.

We first consider the idealized setting where the dataset provides the true rewards $R^*$. Following the policy-improvement argument of \citet{wang2018exponentially}, the reward-weighted policy $\pi^*_{\lambda}(\tau|s) \propto \pi_\beta(\tau|s)\exp\!\left(R^*(\tau)/\lambda\right)$ satisfies $V^{\pi^*_{\lambda}}(s) \geq V^{\pi_\beta}(s)$. Thus, in the absence of reward noise, exponential weighting improves over the behavior policy by shifting probability mass toward higher-reward trajectories.

In practice, however, user feedback is noisy: clicks are affected by presentation effects, ratings vary with user context, and engagement signals contain stochastic variation unrelated to preference. We therefore study the more realistic setting where the observed reward is a noisy estimate of the underlying reward.

\begin{assumption}[Sub-Gaussian Reward Noise]
\label{ass:noise}
The observed reward satisfies $\hat{R}(\tau) = R^*(\tau) + \xi(\tau)$, where $R^*$ is the expected reward and $\xi(\tau)$ is a zero-mean $\sigma$-sub-Gaussian random variable for every trajectory $\tau$.
\end{assumption}

This assumption captures common feedback signals in recommendation systems. Bounded feedback such as clicks or ratings satisfies the assumption through Hoeffding's lemma, while continuous feedback with approximately Gaussian variation is also covered. Importantly, the assumption is imposed only on the final trajectory reward: intermediate rewards may be correlated arbitrarily, as long as their aggregation produces a sub-Gaussian trajectory-level error.

Let $\pi_\lambda(\tau|s) \propto \pi_\beta(\tau|s)\exp\!\left(\hat{R}(\tau)/\lambda\right)$ denote the policy obtained from noisy rewards. Rather than comparing $\pi_\lambda$ only against $\pi_\beta$, we analyze performance against an arbitrary comparator policy $\tilde{\pi}$. This allows the learned policy to be evaluated against any target policy, including an optimal policy, while explicitly accounting for how far that target lies from the offline data distribution.

\begin{theorem}[Policy Improvement Under Noisy Rewards]
\label{thm:noisy_improvement}
Fix a context $s$. Under Assumption~\ref{ass:noise}, with probability at least $1-\delta$, for every comparator policy $\tilde{\pi}$,
\begin{equation}
V^{\pi_\lambda}(s) \;\geq\; V^{\tilde{\pi}}(s)
\;-\; \lambda\, D_{\mathrm{KL}}\!\left( \tilde{\pi}(\cdot|s) \,\|\, \pi_\beta(\cdot|s) \right)
\;-\; 2\epsilon,
\label{eq:comparator_bound}
\end{equation}
where $\epsilon := \sigma\sqrt{2\log(2|\mathcal{T}|/\delta)}$.
\end{theorem}

Theorem~\ref{thm:noisy_improvement} separates the performance gap into two interpretable terms. The first term, $\lambda\, D_{\mathrm{KL}}(\tilde{\pi} \,\|\, \pi_\beta)$, is the \emph{coverage cost}: offline learning cannot reliably imitate policies that place probability mass far outside the logged data. The second term, $2\epsilon$, is the \emph{noise cost} induced by stochastic user feedback. Since $\log |\mathcal{T}| \leq T\log|\mathcal{A}|$, the noise penalty scales as $O(\sigma\sqrt{T\log|\mathcal{A}|})$, which grows only logarithmically with the catalog size.

The previous theorem treats $\lambda$ only as a regularization parameter. When rewards are bounded, we can characterize how $\lambda$ explicitly controls the noise--coverage trade-off.

\begin{theorem}[Temperature-Dependent Comparator Bound]
\label{thm:temp_bound}
Assume $R^*(\tau) \in [0, R_{\max}]$. Under the conditions of Theorem~\ref{thm:noisy_improvement}, with probability at least $1-\delta$, for every comparator policy $\tilde{\pi}$:
\begin{align}
V^{\pi_\lambda}(s) \;\geq\; V^{\tilde{\pi}}(s)
&\;-\; \lambda\, D_{\mathrm{KL}}\!\left( \tilde{\pi}(\cdot|s) \,\|\, \pi_\beta(\cdot|s) \right) \nonumber \\
&\;-\; R_{\max}\!\left( e^{2\epsilon/\lambda} - 1 \right).
\label{eq:temp_comparator_bound}
\end{align}
For $\lambda \geq 2\epsilon$, $\; R_{\max}\!\left(e^{2\epsilon/\lambda} - 1\right) \leq \frac{4R_{\max}\epsilon}{\lambda}$.
\end{theorem}

Theorem~\ref{thm:temp_bound} reveals the role of temperature. Increasing $\lambda$ makes the policy more conservative, increasing the comparator cost linearly while reducing sensitivity to reward noise at rate $O(1/\lambda)$. Conversely, smaller $\lambda$ enables stronger exploitation of high-reward trajectories but amplifies reward noise.

Finally, we choose the comparator policy $\tilde{\pi}$ to be the optimal policy for context $s$, which converts the general comparator bound into a characterization of the suboptimality relative to the highest-reward trajectory. Let $\tau^\star(s) \in \argmax_{\tau} R^*(\tau)$ denote an optimal trajectory from context $s$, and define its coverage under the behavior policy as $p^\star(s) = \pi_\beta(\tau^\star(s)|s)$.

\begin{theorem}[Coverage-Dependent Suboptimality Bound]
\label{thm:coverage_tradeoff}
Under the conditions of Theorem~\ref{thm:temp_bound}, with probability at least $1-\delta$:
\begin{equation}
\max_\tau R^*(\tau) - V^{\pi_\lambda}(s)
\leq \lambda \log\frac{1}{p^\star(s)}
+ R_{\max}\!\left( e^{2\epsilon/\lambda} - 1 \right).
\label{eq:coverage_tradeoff}
\end{equation}
For $\lambda \geq 2\epsilon$, the right-hand side is at most $\lambda\log(1/p^\star(s)) + 4R_{\max}\epsilon/\lambda$, which is minimized at
\begin{equation}
\lambda^\star = 2\sqrt{\frac{R_{\max}\,\epsilon}{\log(1/p^\star(s))}},
\label{eq:lambda_star}
\end{equation}
where the minimized bound equals $4\sqrt{R_{\max}\,\epsilon\,\log(1/p^\star(s))}$.
\end{theorem}

Theorem~\ref{thm:coverage_tradeoff} gives the central trade-off of exponential reward weighting. The first term measures the difficulty of recovering the optimal trajectory from offline data: if the behavior policy rarely visits the optimal trajectory, no support-preserving method can reliably recover it. The second term captures sensitivity to noisy feedback. Therefore, the optimal temperature increases with reward noise and decreases with data coverage.

The coverage term admits a useful sequential interpretation. Using the autoregressive factorization from Equation~\eqref{eq:traj_policy},
\begin{equation}
\log\frac{1}{p^\star(s)} = \sum_{t=1}^{T} \log\frac{1}{\pi_\beta(a_t^\star \mid s_t^\star)},
\end{equation}
where $(s_t^\star, a_t^\star)$ are the states and actions along the optimal trajectory. Thus, the coverage penalty decomposes into a sum of per-step costs. Although the probability of observing a specific length-$T$ trajectory under the behavior policy decreases multiplicatively with the horizon, the corresponding penalty in the bound is its logarithm, which accumulates additively across steps. Therefore, the cost of recovering a target trajectory grows linearly with the trajectory length rather than exponentially.

When rewards are noiseless, the noise term disappears and smaller temperatures increasingly favor optimal trajectories, with the remaining limitation determined only by offline coverage. Under noisy rewards, however, $\lambda$ cannot be reduced arbitrarily: aggressive reweighting improves exploitation of high-reward trajectories but amplifies sensitivity to reward noise, while larger temperatures improve robustness at the cost of staying closer to the behavior policy. This trade-off explains the inverted-U relationship between temperature and recommendation performance observed empirically in Section~\ref{sec:experiments} (Figure~\ref{fig:lambda_main}). Although Equation~\eqref{eq:lambda_star} characterizes the optimal temperature $\lambda^*$, evaluating it requires unknown quantities such as the reward noise level and the coverage of the optimal trajectory. Therefore, it cannot be directly used for hyperparameter selection in practice, and we tune $\lambda$ empirically.
\section{Experiments}
\label{sec:experiments}

\textbf{Experimental Setting.}
Our theoretical analysis considers the general trajectory-level setting for autoregressive generative recommenders. In experiments, we instantiate this framework in the next-item recommendation setting, corresponding to the single-step case ($T=1$). In this setting, the trajectory consists of a user context and one recommended item, $\tau=(s,a)$, and the trajectory reward reduces to the observed item-level feedback, $R(\tau)=r(s,a)$. This allows us to directly evaluate exponential reward weighting while preserving the connection to the general formulation.

\begin{table*}[t]
\centering
\caption{Dataset Statistics}
\label{tab:dataset_stats}
\begin{tabular}{lcccc}
\toprule
Dataset & Users & Items & Avg. Interactions/User & Test Cases ($r \geq 4.5$) \\
\midrule
ML-1M & 6,040 & 3,706 & 165.60 & 1,530 \\
ML-20M & 138,493 & 26,744 & 144.41 & 40,667 \\
Amazon Books & 811,227 & 695,762 & 14.47 & 424,271 \\
StreamCo & O(Millions) & O(Thousands) & O(Tens) & O(Thousands) \\
\bottomrule
\end{tabular}
\end{table*}

\textbf{Datasets.}
We evaluate on three public sequential recommendation benchmarks and one large-scale industrial dataset: MovieLens 1M (ML-1M), MovieLens 20M (ML-20M)~\cite{harper2015movielens}, Amazon Reviews~\cite{mcauley2015image}, and an anonymized proprietary streaming dataset, denoted as \emph{StreamCo}. Each dataset contains user-item interaction sequences with associated feedback signals. Table~\ref{tab:dataset_stats} summarizes the statistics of each dataset. The benchmarks cover a range of recommendation settings, from dense movie recommendation to sparse large-catalog recommendation.

\textbf{Model and Training Setup.}
We use HSTU~\cite{zhai2024actions} as the underlying generative recommender model. The pretrained HSTU model serves as the behavior cloning (BC) baseline, predicting the next item given the user's interaction history.
For reward-model-based methods, we train a lightweight reward head on top of the recommender representation to predict the feedback associated with the selected item. The predicted item-level reward is used as the optimization signal for methods that require a learned reward function.

\textbf{Algorithms.}
We compare five approaches spanning behavior cloning, reward-weighted supervised fine-tuning, and RL-based alignment:
\begin{itemize}
    \item \textbf{BC} (Behavior Cloning): The pretrained HSTU model~\cite{zhai2024actions}, trained only with next-item prediction and without reward optimization.
    \item \textbf{RSFT} (Reward-weighted Supervised Fine-Tuning): A reward-weighted supervised fine-tuning approach that weights training examples according to their observed rewards~\cite{mukherjeeoffline}.
    \item \textbf{DPO} (Direct Preference Optimization): A preference optimization method that increases the likelihood of preferred items relative to dispreferred ones~\cite{rafailov2023direct}. Since recommendation datasets do not contain explicit preference pairs, we construct preferences by sampling candidate items and comparing them using reward-model scores, following online DPO-style training~\cite{deng2025onerec}.
    \item \textbf{PPO} (Proximal Policy Optimization): A reinforcement learning approach~\cite{schulman2017proximal, stiennon2020learning, ouyang2022training} that updates the recommender policy using rewards from the learned reward model.
    \item \textbf{Exp-RSFT} (Exponential Reward-weighted Supervised Fine-Tuning): Our proposed method, described in Algorithm~\ref{alg:exp-rw-sft}. In the experimental $T=1$ setting, it weights next-item likelihoods by $\exp(r(s,a)/\lambda)$, which is the single-step instantiation of the trajectory-level exponential weighting derived in Section~\ref{sec:exp_rw}.
\end{itemize}

Due to the proprietary nature of the StreamCo dataset, we report results relative to RSFT rather than absolute performance values. We train the algorithms for 30, 12, and 10 epochs on ML-1M, ML-20M, and Amazon Books, respectively.

\textbf{Evaluation.}
We evaluate on the standard test splits of ML-1M, ML-20M, and Amazon Books, restricting evaluation to held-out interactions with rating $\geq 4.5$. This focuses evaluation on high-preference recommendations and aligns with the reward optimization objective. All methods, including BC, are evaluated on the same filtered test sets.
For each test case, let $r_i$ denote the rank assigned to the ground-truth item among all candidate items. We report Hit Rate (HR@$K$), Normalized Discounted Cumulative Gain (NDCG@$K$), and Mean Reciprocal Rank (MRR):
\begin{align}
    \text{HR@}K &= \frac{1}{N}\sum_{i=1}^{N}\mathbbm{1}(r_i \leq K),\\
    \text{NDCG@}K &= \frac{1}{N}\sum_{i=1}^{N}
    \frac{\mathbbm{1}(r_i \leq K)}
    {\log_2(r_i+1)},\\
    \text{MRR} &= \frac{1}{N}\sum_{i=1}^{N}\frac{1}{r_i}.
\end{align}
We report HR and NDCG at cutoffs $K\in\{10,50\}$.

\begin{table}[!ht]
  \centering
  \caption{Reward model prediction performance compared to naive baselines. The trained reward model fails to outperform the simple Item Mean baseline on most metrics. Columns abbreviate the User-Mean, Item-Mean, Global-Mean, and trained reward-model predictors. Best results per row are \textbf{bolded}.}
  \label{tab:rm_baselines}
  \small
  \setlength{\tabcolsep}{3.5pt}
  \begin{tabular}{llcccc}
    \toprule
    Dataset & Metric & User & Item & Global & RM \\
    \midrule
    \multirow{2}{*}{ML-1M}
      & MSE & 1.0784 & \textbf{0.9179} & 1.2394 & 1.2166 \\
      & MAE & 0.8232 & \textbf{0.7619} & 0.9268 & 0.8680 \\
    \midrule
    \multirow{2}{*}{ML-20M}
      & MSE & 0.9550 & \textbf{0.8898} & 1.0939 & 0.8963 \\
      & MAE & 0.7588 & 0.7320 & 0.8423 & \textbf{0.7307} \\
    \midrule
    \multirow{2}{*}{Amazon Books}
      & MSE & 1.1578 & \textbf{0.7290} & 1.1868 & 1.1067 \\
      & MAE & 0.7411 & \textbf{0.5651} & 0.8636 & 0.8140 \\
    \midrule
    \multirow{2}{*}{StreamCo}
      & MSE & \textbf{2.412} & 2.568 & 2.568 & 2.486 \\
      & MAE & 2.112 & 2.442 & 2.451 & \textbf{1.916} \\
    \bottomrule
  \end{tabular}
\end{table}

\begin{table*}[t]
  \centering
  \caption{Recommendation performance on the public benchmarks (rating $\geq 4.5$) and the proprietary StreamCo dataset (high-reward trajectories). Best results per dataset and metric are \textbf{bolded}. StreamCo values are relative changes from RSFT, used in place of BC due to data confidentiality.}
  \label{tab:public_results}
  \begin{tabular}{llcccccc}
    \toprule
    Dataset & Method & NDCG@10 & NDCG@50 & HR@10 & HR@50 & MRR & Avg Reward \\
    \midrule
    \multirow{5}{*}{ML-1M}
      & BC & 0.1304 & 0.1866 & 0.2405 & 0.4967 & 0.1116 & 3.8672 \\
      & DPO & 0.0091 & 0.0231 & 0.0163 & 0.0824 & 0.0122 & 3.3737 \\
      & PPO & 0.0303 & 0.0617 & 0.0673 & 0.2124 & 0.0285 & \textbf{4.4775} \\
      & RSFT & 0.1437 & 0.2003 & 0.2608 & 0.5163 & 0.1230 & 3.6558 \\
      & Exp-RSFT & \textbf{0.1465} & \textbf{0.2029} & \textbf{0.2712} & \textbf{0.5261} & \textbf{0.1235} & 3.6933 \\
    \midrule
    \multirow{5}{*}{ML-20M}
      & BC & 0.1753 & 0.2313 & 0.3035 & 0.5569 & 0.1510 & 3.7388 \\
      & DPO & 0.0196 & 0.0342 & 0.0372 & 0.1052 & 0.0195 & \textbf{4.4008} \\
      & PPO & 0.1164 & 0.1688 & 0.2141 & 0.4536 & 0.1008 & 3.7547 \\
      & RSFT & 0.1847 & 0.2403 & 0.3161 & 0.5677 & 0.1593 & 3.7904 \\
      & Exp-RSFT & \textbf{0.1912} & \textbf{0.2462} & \textbf{0.3238} & \textbf{0.5726} & \textbf{0.1651} & 3.8611 \\
    \midrule
    \multirow{5}{*}{Amazon Books}
      & BC & 0.0328 & 0.0478 & 0.0589 & 0.1279 & 0.0296 & 4.2824 \\
      & DPO & 0.0008 & 0.0020 & 0.0019 & 0.0075 & 0.0013 & \textbf{4.8615} \\
      & PPO & 0.0094 & 0.0156 & 0.0184 & 0.0469 & 0.0088 & 4.4549 \\
      & RSFT & 0.0349 & 0.0514 & 0.0630 & 0.1388 & 0.0314 & 4.3350 \\
      & Exp-RSFT & \textbf{0.0356} & \textbf{0.0520} & \textbf{0.0641} & \textbf{0.1397} & \textbf{0.0319} & 4.3374 \\
    \midrule
    \multirow{4}{*}{StreamCo}
      & RSFT & 0.00\% & 0.00\% & 0.00\% & 0.00\% & 0.00\% & 0.00\% \\
      & DPO & -99.30\% & -96.77\% & -99.14\% & -95.52\% & -95.34\% & \textbf{19.21\%} \\
      & PPO & -4.88\% & -24.01\% & -14.66\% & -35.71\% & -8.60\% & 13.39\% \\
      & Exp-RSFT & \textbf{106.20\%} & \textbf{69.71\%} & \textbf{98.07\%} & \textbf{49.58\%} & \textbf{86.31\%} & 17.27\% \\
    \bottomrule
  \end{tabular}
\end{table*}

\begin{figure*}[!t]
\centering
\includegraphics[width=0.95\textwidth]{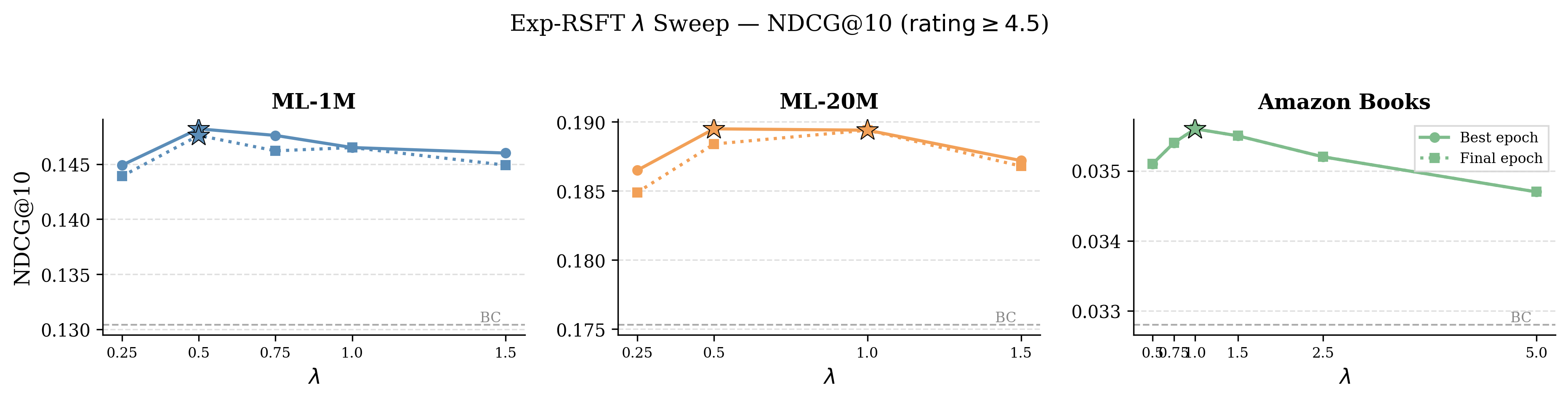}
\caption{NDCG@10 for different values of $\lambda$ across datasets.}
\label{fig:lambda_main}
\end{figure*}

\textbf{Evidence for Reward Hacking.}
Table~\ref{tab:public_results} shows that PPO and DPO can attain high reward-model scores while substantially degrading recommendation quality. This discrepancy suggests reward hacking: optimizing an imperfect learned reward function may exploit weaknesses in the reward model rather than improve alignment with true user preferences. We investigate this phenomenon through two analyses.

\textit{Benchmarking reward model quality.}
To evaluate the reliability of the learned reward model, we compare its prediction error against three simple baselines: User Mean, Item Mean, and Global Mean predictors (Table~\ref{tab:rm_baselines}). The reward model does not consistently outperform these baselines across datasets, indicating that accurately predicting user feedback for unseen items remains challenging.

\textit{Reward model-as-judge evaluation.}
We further test whether optimization methods exploit the reward model by measuring the average predicted reward across 10 generations from each of 10,000 evaluation contexts, reported as ``Avg Reward'' in Table~\ref{tab:public_results}. Reward-model score does not correlate with recommendation quality: on every dataset, the method achieving the highest Avg Reward also experiences severe degradation in ranking metrics (PPO on ML-1M; DPO on ML-20M, Amazon Books, and StreamCo), providing evidence of reward hacking. Training curves in Appendix~\ref{app:metric_evolution} show that this collapse occurs early in training.

\textbf{Robustness--Improvement Tradeoff with $\lambda$.}
Theorems~\ref{thm:temp_bound} and~\ref{thm:coverage_tradeoff} predict a tradeoff controlled by the temperature $\lambda$: small values aggressively emphasize high-reward actions but increase sensitivity to reward noise, while large values remain closer to the behavior policy and reduce potential improvement. We evaluate this prediction by sweeping $\lambda$ across all datasets.
Figure~\ref{fig:lambda_main} shows NDCG@10 as a function of $\lambda$; other metrics exhibit similar trends. Across all datasets, performance follows an inverted-U pattern, improving as $\lambda$ decreases from large values, reaching a maximum at an intermediate temperature ($\lambda \approx 0.5$--$1.0$), and degrading when $\lambda$ becomes too small. This behavior is consistent with the theoretical prediction: large temperatures suppress improvement by limiting deviation from the behavior policy, while excessively small temperatures amplify noise sensitivity. The same trend is observed when selecting either the best or final training epoch (Appendix~\ref{app:lambda-sweep}), ruling out early stopping as an explanation. These results confirm that moderate reward weighting provides the best balance between reward exploitation and robustness to noisy feedback.

\section{Conclusions and Limitations}
We analyzed exponential reward-weighted fine-tuning for generative recommenders trained offline from logged user feedback. Our theoretical analysis shows that the suboptimality of the learned policy decomposes into two interpretable costs: a coverage cost determined by limitations of the logging policy, and a noise cost arising from imperfect feedback. The temperature $\lambda$ balances these competing effects, yielding an interior optimum that explains the observed tradeoff between reward exploitation and robustness. Experiments on three public benchmarks and a large-scale proprietary dataset validate these predictions: temperature sweeps exhibit the predicted inverted-U behavior, while standard reward-model-based post-training methods can over-optimize imperfect reward signals and degrade recommendation quality. Exp-RSFT provides a simple and effective alternative for post-training generative recommenders in sparse, noisy offline settings.

Our analysis also has several limitations. First, it focuses on the regime where users interact with only a small fraction of a large catalog and feedback is available as scalar rewards. In settings with substantially richer supervision, such as dense preference comparisons or extensive user feedback coverage, other post-training approaches may become more effective. Second, Exp-RSFT does not eliminate the fundamental challenge of preference estimation under sparse interactions. More sophisticated reward modeling approaches, including uncertainty-aware or conservative reward models and models incorporating additional user signals, may alleviate some of the limitations observed in our experiments. However, when feedback covers only a small portion of the catalog, generalizing reliable preferences to unseen items remains challenging. Third, our theoretical guarantees assume sub-Gaussian reward noise; real-world feedback may exhibit heavier tails or more complex dependencies. Extending the analysis to broader noise models and studying the interaction between improved reward modeling and offline policy optimization are important directions for future work.

\bibliography{main}

\clearpage
\appendix
\section{Proofs}
\label{sec:proofs}

\setcounter{equation}{0}
\renewcommand{\theequation}{A\arabic{equation}}

Throughout the appendix, fix a context $s$ and recall the notation of Section~\ref{sec:theory}: $\mathcal{T}$ is the set of trajectories reachable from $s$; $\hat{R} = R^* + \xi$ are the noisy rewards of Assumption~\ref{ass:noise}; $\pi_\lambda(\tau|s) \propto \pi_\beta(\tau|s)\exp(\hat{R}(\tau)/\lambda)$ is the exponential tilt computed from the noisy rewards; and values $V^{\pi}(s) = \mathbb{E}_{\tau \sim \pi(\cdot|s)}[R^*(\tau)]$ are always evaluated under the true rewards. If $D_{\mathrm{KL}}(\tilde{\pi}(\cdot|s) \,\|\, \pi_\beta(\cdot|s)) = \infty$, every bound below holds trivially, so we assume throughout that it is finite; equivalently, $\tilde{\pi}(\cdot|s)$ is supported on the support of $\pi_\beta(\cdot|s)$. We also note that the noiseless monotonic improvement used as a warm-up in Section~\ref{sec:theory} is the $\sigma = 0$ case of Theorem~\ref{thm:noisy_improvement}: with $\sigma = 0$ we have $\epsilon = 0$, and taking $\tilde{\pi} = \pi_\beta$ gives $V^{\pi^*_\lambda}(s) \geq V^{\pi_\beta}(s)$ deterministically.

\subsection{Proof of Theorem~\ref{thm:noisy_improvement}}
\label{subsec:noisy-improvement-proof}

\noindent\textbf{Theorem~\ref{thm:noisy_improvement}} (Policy Improvement Under Noisy Rewards; restated)\textbf{.}
\emph{Fix a context $s$. Under Assumption~\ref{ass:noise}, with probability at least $1 - \delta$, for every comparator policy $\tilde{\pi}$:
\begin{equation*}
    V^{\pi_\lambda}(s) \;\geq\; V^{\tilde{\pi}}(s)
    \;-\; \lambda\, D_{\mathrm{KL}}\!\left(\tilde{\pi}(\cdot|s)\,\|\,
    \pi_\beta(\cdot|s)\right) \;-\; 2\epsilon,
\end{equation*}
where $\epsilon := \sigma\sqrt{2\log(2|\mathcal{T}|/\delta)}$.}

\begin{proof}
\textbf{Step 1 (noise event).} Since each $\xi(\tau)$ is $\sigma$-sub-Gaussian, a standard tail bound combined with a union bound over the $|\mathcal{T}|$ trajectories yields, with probability at least $1-\delta$,
\begin{equation}\label{eq:noise_event}
|\xi(\tau)| \leq \epsilon \quad \forall \tau \in \mathcal{T},
\qquad \epsilon := \sigma\sqrt{2\log(2|\mathcal{T}|/\delta)}.
\end{equation}
No independence across trajectories is needed: the union bound uses only the marginal tail of each $\xi(\tau)$. We condition on this event for the remainder of the proof and fix a comparator policy $\tilde{\pi}$.

\textbf{Step 2 (variational identity).} Define
\begin{equation}\label{eq:F_def}
F(\pi) = \sum_{\tau} \pi(\tau|s)\,\hat{R}(\tau)
- \lambda\, D_{\mathrm{KL}}\!\left(\pi(\cdot|s) \,\|\,
\pi_\beta(\cdot|s)\right),
\end{equation}
and recall that the tilted policy normalizes as
\begin{gather*}
\pi_\lambda(\tau|s) = \frac{1}{Z}\,\pi_\beta(\tau|s)\exp\!\left(\frac{\hat{R}(\tau)}{\lambda}\right), \\
Z = \sum_{\tau} \pi_\beta(\tau|s)\exp\!\left(\frac{\hat{R}(\tau)}{\lambda}\right).
\end{gather*}
For any policy $\pi$ with $\pi(\cdot|s)$ supported on the support of $\pi_\beta(\cdot|s)$,
\begin{align*}
D_{\mathrm{KL}}(\pi \,\|\, \pi_\lambda)
&= \sum_{\tau} \pi(\tau|s)\log\frac{\pi(\tau|s)}
{\pi_\beta(\tau|s)\exp(\hat{R}(\tau)/\lambda)/Z} \\
&= D_{\mathrm{KL}}(\pi \,\|\, \pi_\beta) + \log Z \\
&\quad - \frac{1}{\lambda}\sum_{\tau} \pi(\tau|s)\,\hat{R}(\tau),
\end{align*}
and rearranging yields the identity
\begin{equation}\label{eq:variational_identity}
F(\pi) = \lambda \log Z - \lambda\, D_{\mathrm{KL}}(\pi \,\|\,
\pi_\lambda).
\end{equation}
Since the KL divergence is nonnegative and vanishes iff
$\pi = \pi_\lambda$, the tilted policy is the unique maximizer of $F$; in
particular, $F(\pi_\lambda) \geq F(\tilde{\pi})$.

\textbf{Step 3 (comparator inequality).} Writing out
$F(\pi_\lambda) \geq F(\tilde{\pi})$ via \eqref{eq:F_def} and dropping the nonnegative term
$\lambda D_{\mathrm{KL}}(\pi_\lambda \,\|\, \pi_\beta) \geq 0$ on the
left-hand side:
\begin{equation}\label{eq:comparator_dominance}
\sum_{\tau} \pi_\lambda(\tau|s)\,\hat{R}(\tau)
\;\geq\; \sum_{\tau} \tilde{\pi}(\tau|s)\,\hat{R}(\tau)
- \lambda\, D_{\mathrm{KL}}(\tilde{\pi} \,\|\, \pi_\beta).
\end{equation}

\textbf{Step 4 (noise conversion).} On the event \eqref{eq:noise_event}, since
$\pi_\lambda(\cdot|s)$ and $\tilde{\pi}(\cdot|s)$ are probability
distributions,
\begin{gather*}
V^{\pi_\lambda}(s) \geq \sum_{\tau} \pi_\lambda(\tau|s)\,\hat{R}(\tau) - \epsilon, \\
\sum_{\tau} \tilde{\pi}(\tau|s)\,\hat{R}(\tau) \geq V^{\tilde{\pi}}(s) - \epsilon.
\end{gather*}
Chaining these with \eqref{eq:comparator_dominance} completes the proof.
\end{proof}

\subsection{Proof of Theorem~\ref{thm:temp_bound}}
\label{subsec:temp-bound-proof}

\noindent\textbf{Theorem~\ref{thm:temp_bound}} (Temperature-Dependent Comparator Bound; restated)\textbf{.}
\emph{Under the conditions of Theorem~\ref{thm:noisy_improvement}, if rewards are bounded $R^*(\tau) \in [0, R_{\max}]$, then with probability at least $1-\delta$, for every comparator policy $\tilde{\pi}$:
\begin{align*}
    V^{\pi_\lambda}(s) \;\geq\; V^{\tilde{\pi}}(s)
    &\;-\; \lambda\, D_{\mathrm{KL}}\!\left(\tilde{\pi}(\cdot|s)\,\|\,
    \pi_\beta(\cdot|s)\right) \\
    &\;-\; R_{\max}\!\left(e^{2\epsilon/\lambda} - 1\right).
\end{align*}
When $\lambda \geq 2\epsilon$, the noise cost simplifies: $R_{\max}\left(e^{2\epsilon/\lambda} - 1\right) \leq 4R_{\max}\epsilon/\lambda$.}

\begin{proof}
Let $\pi_\lambda^*$ denote the tilt under the \emph{true} rewards and $\pi_\lambda$ the tilt under the noisy rewards:
\begin{align}
\pi_\lambda^*(\tau|s) &= \frac{1}{Z^*}\,\pi_\beta(\tau|s)\exp\!\left(\frac{R^*(\tau)}{\lambda}\right), \label{eq:clean_tilt}\\
\pi_\lambda(\tau|s) &= \frac{1}{Z}\,\pi_\beta(\tau|s)\exp\!\left(\frac{\hat{R}(\tau)}{\lambda}\right), \label{eq:noisy_tilt}
\end{align}
with normalizing factors $Z^* = \sum_{\tau} \pi_\beta(\tau|s)\exp(R^*(\tau)/\lambda)$ and $Z = \sum_{\tau} \pi_\beta(\tau|s)\exp(\hat{R}(\tau)/\lambda)$. As in the proof of Theorem~\ref{thm:noisy_improvement}, we condition on the event \eqref{eq:noise_event}, which holds with probability at least $1 - \delta$, and fix a comparator policy $\tilde{\pi}$.

\textbf{Step 1 (clean comparator inequality).} Applying the variational identity \eqref{eq:variational_identity} with the true rewards $R^*$ in place of $\hat{R}$ shows that $\pi_\lambda^*$ is the unique maximizer of $F^*(\pi) = \sum_{\tau} \pi(\tau|s)\,R^*(\tau) - \lambda\, D_{\mathrm{KL}}(\pi \,\|\, \pi_\beta)$. Hence $F^*(\pi_\lambda^*) \geq F^*(\tilde{\pi})$, and dropping the nonnegative term $\lambda D_{\mathrm{KL}}(\pi_\lambda^* \,\|\, \pi_\beta) \geq 0$ on the left-hand side gives, deterministically,
\begin{equation}\label{eq:clean_comparator}
    V^{\pi_\lambda^*}(s) \;\geq\; V^{\tilde{\pi}}(s) - \lambda\, D_{\mathrm{KL}}(\tilde{\pi} \,\|\, \pi_\beta).
\end{equation}
No noise event is needed here: both sides involve only the true rewards.

\textbf{Step 2 (decomposition).} The value of the noisy tilt decomposes as
\begin{equation}\label{eq:decomp}
    V^{\pi_\lambda}(s) = V^{\pi_\lambda^*}(s) - \left(V^{\pi_\lambda^*}(s) - V^{\pi_\lambda}(s)\right),
\end{equation}
so by \eqref{eq:clean_comparator} it suffices to bound the noise difference $|V^{\pi_\lambda^*}(s) - V^{\pi_\lambda}(s)|$.

\textbf{Step 3 (bounding the noise difference).} Since $R^*(\tau) \in [0, R_{\max}]$:
\begin{align}
    \left|V^{\pi_\lambda^*}(s) - V^{\pi_\lambda}(s)\right|
    &= \left|\sum_{\tau} \left(\pi_\lambda^*(\tau|s) - \pi_\lambda(\tau|s)\right) R^*(\tau)\right| \nonumber \\
    &\leq R_{\max} \sum_{\tau} \left|\pi_\lambda^*(\tau|s) - \pi_\lambda(\tau|s)\right| \nonumber \\
    &= 2R_{\max}\,  D_{\mathrm{TV}}\!\left(\pi_\lambda^*(\cdot|s),\, \pi_\lambda(\cdot|s)\right). \label{eq:tv_reduction}
\end{align}
Define the noise weight $w(\tau) = \exp(\xi(\tau)/\lambda)$. On the event \eqref{eq:noise_event}:
\begin{equation}\label{eq:w_bounds}
    e^{-\epsilon/\lambda} \leq w(\tau) \leq e^{\epsilon/\lambda} \quad \forall\, \tau \in \mathcal{T}.
\end{equation}
Since $\hat{R}(\tau) = R^*(\tau) + \xi(\tau)$, the noisy partition function factorizes as
\begin{align}
    Z &= \sum_{\tau} \pi_\beta(\tau|s)\exp\!\left(\frac{R^*(\tau)}{\lambda}\right) w(\tau) \nonumber \\
    &= Z^* \sum_{\tau} \pi_\lambda^*(\tau|s)\, w(\tau), \label{eq:Z_factor}
\end{align}
where the last step uses \eqref{eq:clean_tilt}. The ratio $Z/Z^*$ is therefore the expected noise weight under the clean policy; combined with \eqref{eq:w_bounds}, this gives
\begin{equation}\label{eq:rho_bounds}
    e^{-\epsilon/\lambda} \leq \rho \leq e^{\epsilon/\lambda}, \quad \text{where } \rho := Z^*/Z.
\end{equation}
For any trajectory $\tau$, dividing \eqref{eq:noisy_tilt} by \eqref{eq:clean_tilt} gives the policy ratio
\begin{equation}\label{eq:ratio}
    \frac{\pi_\lambda(\tau|s)}{\pi_\lambda^*(\tau|s)} = w(\tau)\,\rho.
\end{equation}
Substituting \eqref{eq:ratio} into the total variation:
\begin{equation}\label{eq:tv_expr}
    2\,D_{\mathrm{TV}}\!\left(\pi_\lambda^*,\, \pi_\lambda\right)
    = \sum_{\tau} \pi_\lambda^*(\tau|s)\left|w(\tau)\,\rho - 1\right|.
\end{equation}
From \eqref{eq:w_bounds} and \eqref{eq:rho_bounds}, the product $w(\tau)\,\rho$ lies in $[e^{-2\epsilon/\lambda},\, e^{2\epsilon/\lambda}]$, so
\begin{equation}
    |w(\tau)\,\rho - 1| \leq \max\!\left(e^{2\epsilon/\lambda} - 1,\; 1 - e^{-2\epsilon/\lambda}\right) = e^{2\epsilon/\lambda} - 1,
\end{equation}
where the equality holds because $e^x - 1 \geq 1 - e^{-x}$ for $x \geq 0$. Substituting back into \eqref{eq:tv_expr}:
\begin{equation}\label{eq:tv_bound}
    D_{\mathrm{TV}}\!\left(\pi_\lambda^*(\cdot|s),\, \pi_\lambda(\cdot|s)\right) \leq \frac{e^{2\epsilon/\lambda} - 1}{2}.
\end{equation}

\textbf{Step 4 (combining).} Substituting \eqref{eq:tv_bound} into \eqref{eq:tv_reduction}, and the result together with \eqref{eq:clean_comparator} into \eqref{eq:decomp}:
\begin{align}
    V^{\pi_\lambda}(s) \geq V^{\tilde{\pi}}(s)
    &- \lambda\, D_{\mathrm{KL}}(\tilde{\pi} \,\|\, \pi_\beta) \nonumber \\
    &- R_{\max}\!\left(e^{2\epsilon/\lambda} - 1\right).
\end{align}
When $\lambda \geq 2\epsilon$, we have $2\epsilon/\lambda \leq 1$, and $e^x - 1 \leq 2x$ for $x \in [0,1]$, giving $R_{\max}(e^{2\epsilon/\lambda} - 1) \leq 4R_{\max}\epsilon/\lambda$. Substituting $\epsilon = \sigma\sqrt{2\log(2|\mathcal{T}|/\delta)}$ completes the proof.
\end{proof}

\subsection{Proof of Theorem~\ref{thm:coverage_tradeoff}}
\label{subsec:coverage-proof}

\noindent\textbf{Theorem~\ref{thm:coverage_tradeoff}} (Coverage-Dependent Suboptimality Bound; restated)\textbf{.}
\emph{Under the conditions of Theorem~\ref{thm:temp_bound}, with probability at least $1-\delta$:
\begin{equation*}
    \max_{\tau} R^*(\tau) - V^{\pi_\lambda}(s)
    \leq \lambda \log\frac{1}{p^\star(s)}
    + R_{\max}\!\left(e^{2\epsilon/\lambda} - 1\right).
\end{equation*}
When $\lambda \geq 2\epsilon$, the right-hand side is at most $\lambda\log(1/p^\star) + 4R_{\max}\epsilon/\lambda$, which is minimized at $\lambda^\star = 2\sqrt{R_{\max}\epsilon/\log(1/p^\star)}$, where the minimized bound equals $4\sqrt{R_{\max}\,\epsilon\,\log(1/p^\star)}$.}

\begin{proof}
If $p^\star(s) = 0$ the right-hand side of Equation~\eqref{eq:coverage_tradeoff} is infinite and the bound holds trivially, so assume $p^\star(s) > 0$. Take $\tilde{\pi}(\cdot|s)$ to be the optimal policy from $s$, i.e., the point mass on $\tau^\star(s)$. Then $V^{\tilde{\pi}}(s) = \max_{\tau} R^*(\tau)$ and
\begin{equation}
D_{\mathrm{KL}}\!\left(\tilde{\pi}(\cdot|s) \,\|\,
\pi_\beta(\cdot|s)\right)
= \log\frac{1}{\pi_\beta(\tau^\star(s)|s)} = \log\frac{1}{p^\star(s)}.
\end{equation}
Substituting into Theorem~\ref{thm:temp_bound} and rearranging yields Equation~\eqref{eq:coverage_tradeoff}.

For the second claim, when $\lambda \geq 2\epsilon$ the noise cost is at most $4R_{\max}\epsilon/\lambda$ as in Theorem~\ref{thm:temp_bound}, so the right-hand side is at most
\begin{equation}
g(\lambda) := \lambda\log\frac{1}{p^\star} + \frac{4R_{\max}\epsilon}{\lambda}.
\end{equation}
The function $g$ is convex with
$g'(\lambda) = \log(1/p^\star) - 4R_{\max}\epsilon/\lambda^2$, which vanishes at
\begin{gather}
\lambda^\star = 2\sqrt{\frac{R_{\max}\,\epsilon}{\log(1/p^\star)}}, \\
g(\lambda^\star) = 4\sqrt{R_{\max}\,\epsilon\,\log(1/p^\star)}.
\end{gather}
\end{proof}

\section{Additional Experimental Results}
\label{sec:additional-experiments}

\renewcommand{\dbltopfraction}{0.95}
\renewcommand{\topfraction}{0.95}
\renewcommand{\dblfloatpagefraction}{0.5}
\renewcommand{\floatpagefraction}{0.5}

\subsection{Training Details}
\label{app:training-details}
All experiments were run on 4 NVIDIA A100 GPUs. Each reported result corresponds to a single training run per algorithm--dataset configuration; the consistency of Exp-RSFT's improvements across all five evaluation metrics and all four datasets serves as evidence of the robustness of the reported gains. The temperature $\lambda$ and training epochs are specified in Section~\ref{sec:experiments}; all remaining hyperparameter settings, including the fixed random seeds used for training, are provided in the accompanying code, allowing exact replication of the reported results.

\subsection{Rating Distribution}
\label{app:rating-dist}
Figure~\ref{fig:rating_dist} shows the distribution of ratings across the three public datasets. Ratings concentrate at the high end of the scale, which motivates the rating $\geq 4.5$ filter used in our evaluation (Section~\ref{sec:experiments}).

\begin{figure*}[t]
  \centering
  \includegraphics[width=0.72\textwidth]{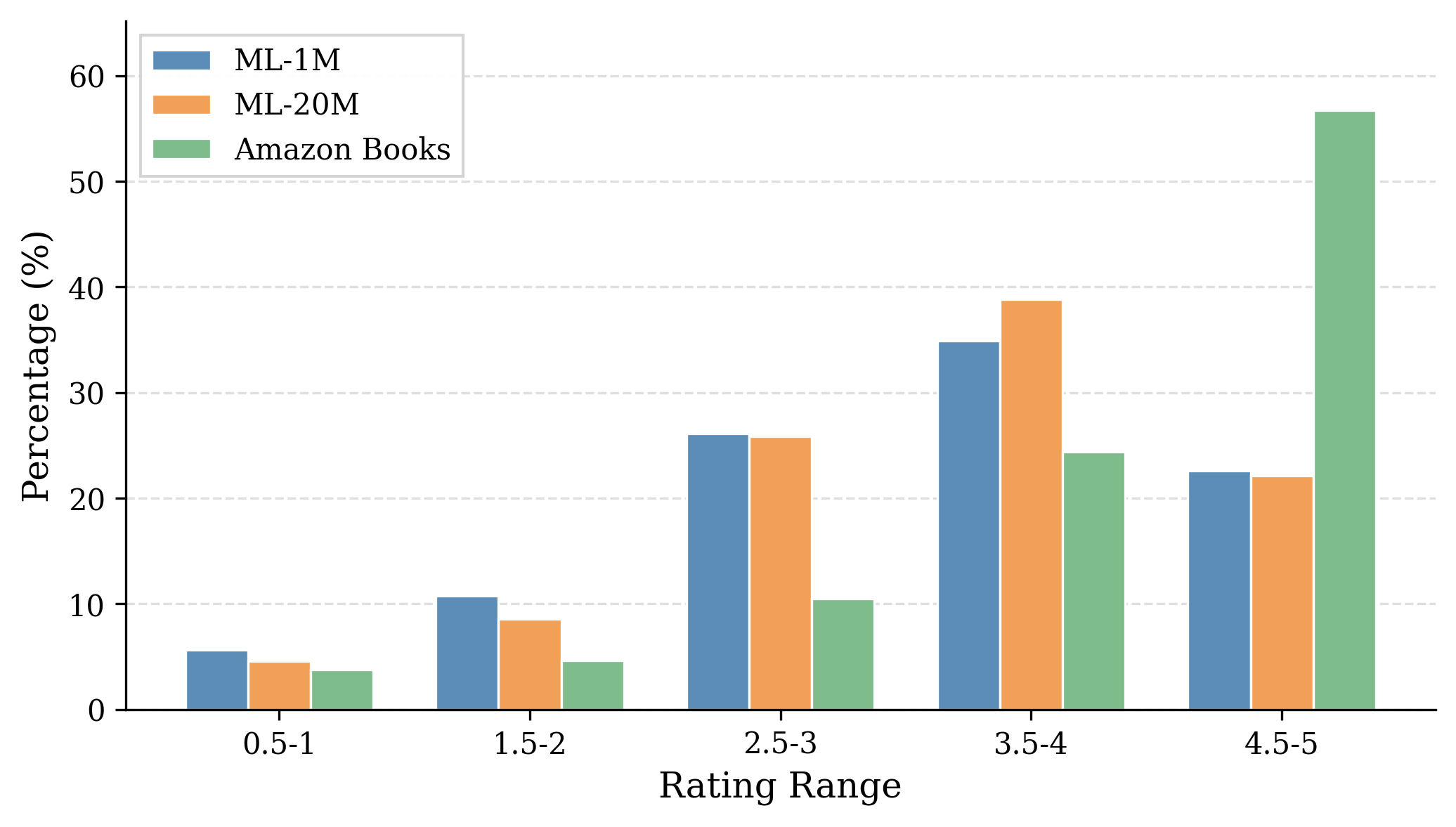}
  \caption{Distribution of ratings in the three public test datasets.}
  \label{fig:rating_dist}
\end{figure*}

\subsection{Temperature Sweep at Best and Final Epochs}
\label{app:lambda-sweep}
Figures~\ref{fig:lambda-sweep-best} and~\ref{fig:lambda-sweep-last} report the full temperature sweep of Section~\ref{sec:experiments} for all five evaluation metrics, selecting the best training epoch and the final training epoch, respectively. The inverted-U shape predicted by Theorem~\ref{thm:coverage_tradeoff} appears in both, ruling out early stopping as a confound.

\begin{figure*}[t]
  \centering
  \includegraphics[width=0.72\textwidth]{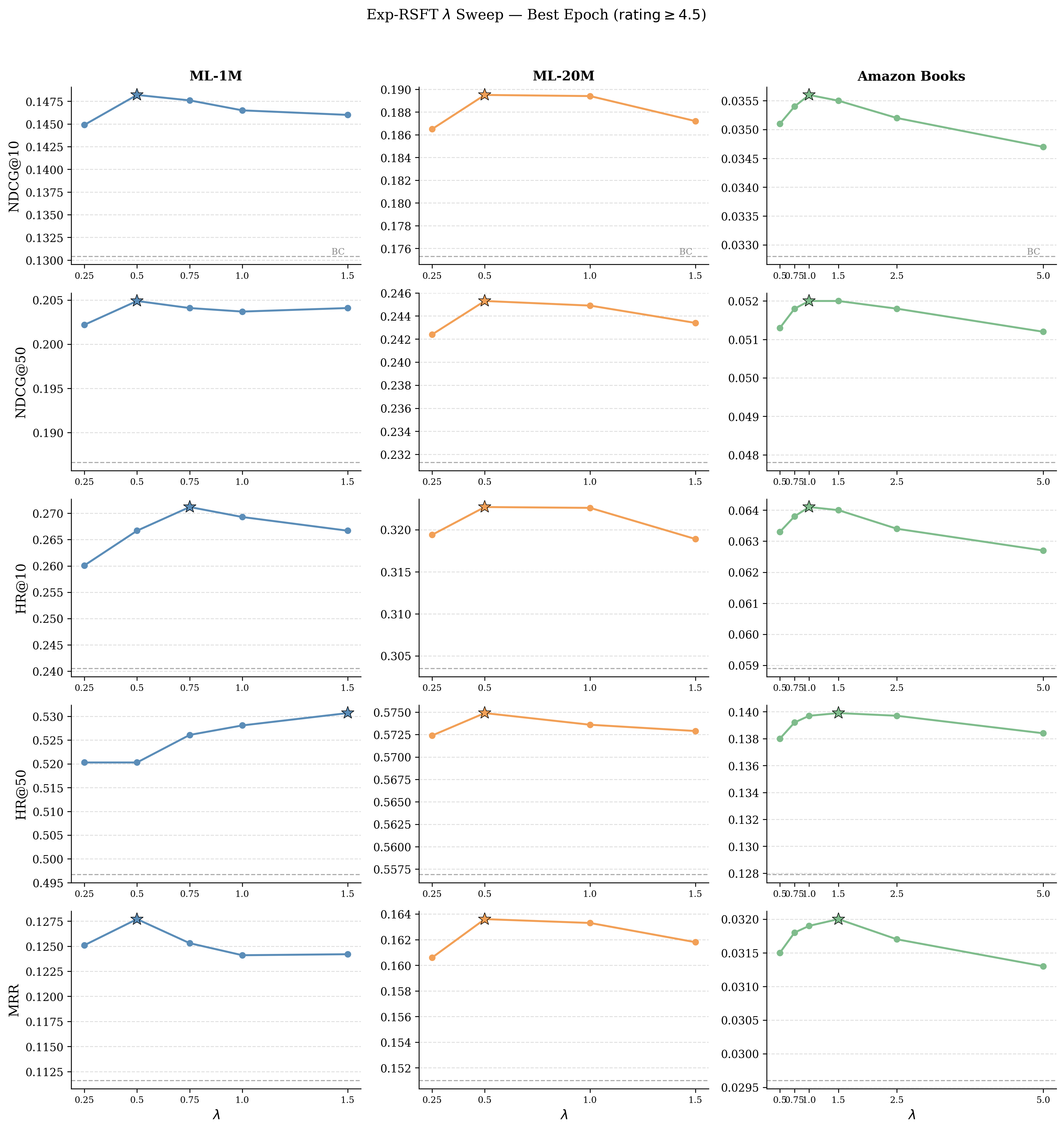}
  \caption{Best-epoch values of all five metrics for different $\lambda$, on the three public benchmarks.}
  \label{fig:lambda-sweep-best}
\end{figure*}

\begin{figure*}[t]
  \centering
  \includegraphics[width=0.72\textwidth]{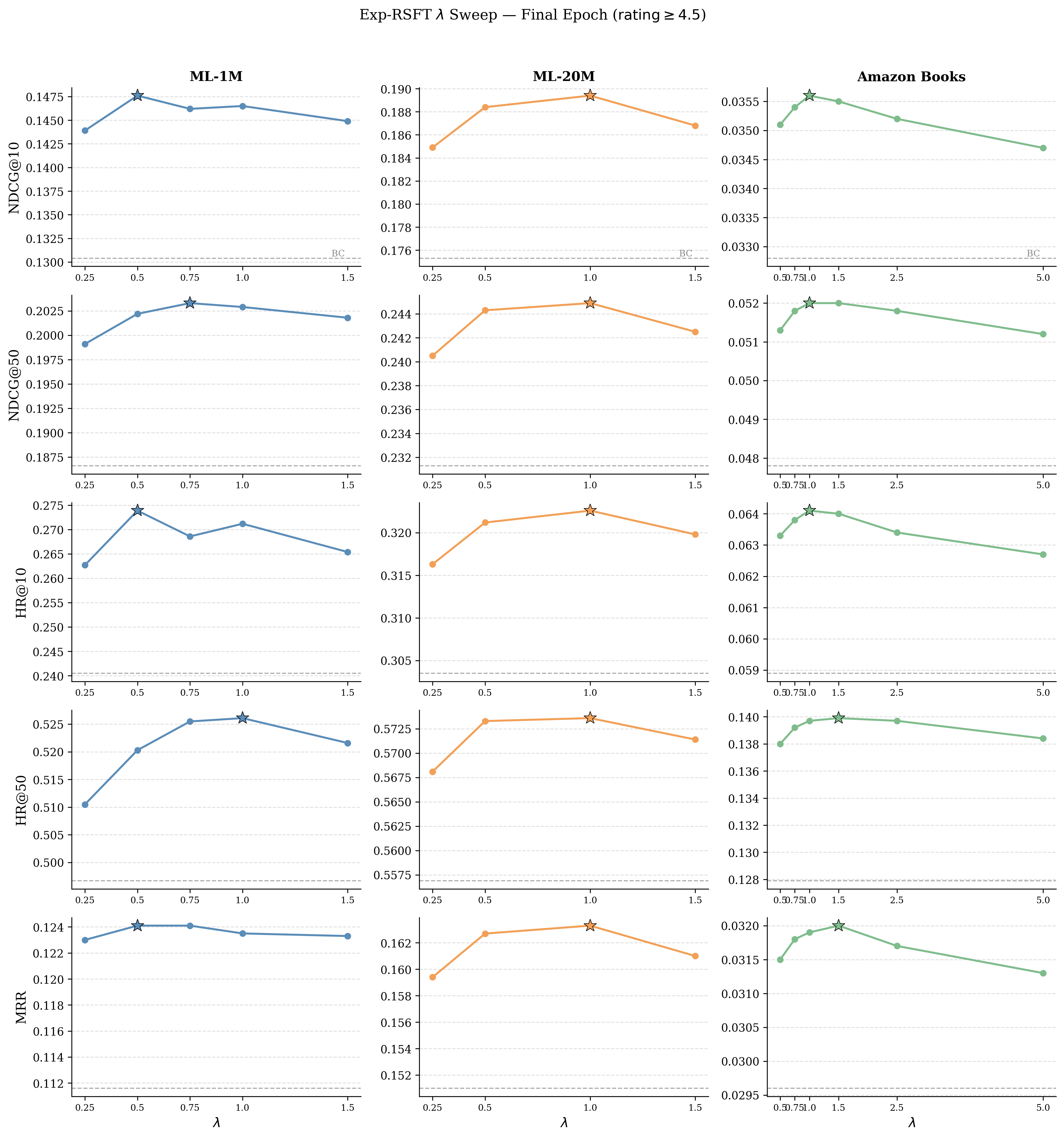}
  \caption{Final-epoch values of all five metrics for different $\lambda$, on the three public benchmarks.}
  \label{fig:lambda-sweep-last}
\end{figure*}

\subsection{Metric Evolution over Training}
\label{app:metric_evolution}
Figures~\ref{fig:metrics_all_ml-1m}--\ref{fig:metrics_all_amazon-books} show the evolution of all five metrics across training epochs for every algorithm on the three public benchmarks; epoch 0 corresponds to the BC baseline. PPO and DPO collapse early in training, consistent with reward-model over-optimization. Figures~\ref{fig:metrics_zoom_ml-1m}--\ref{fig:metrics_zoom_amazon-books} zoom in on RSFT and Exp-RSFT only: Exp-RSFT improves on RSFT at essentially every epoch on every dataset.

\begin{figure*}[t]
  \centering
  \includegraphics[width=0.72\textwidth]{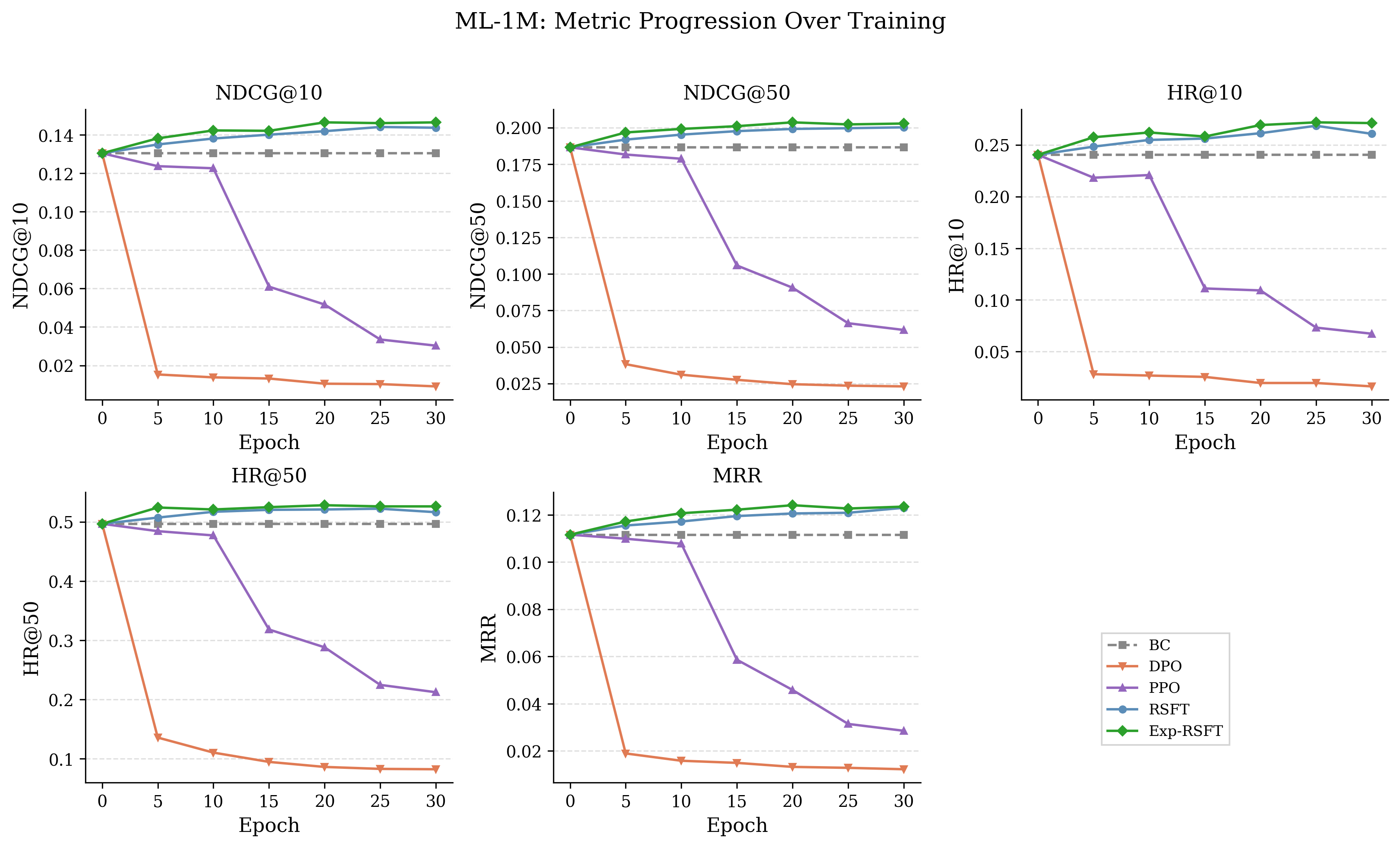}
  \caption{All metrics across training epochs for all algorithms on ML-1M.}
  \label{fig:metrics_all_ml-1m}
\end{figure*}

\begin{figure*}[t]
  \centering
  \includegraphics[width=0.72\textwidth]{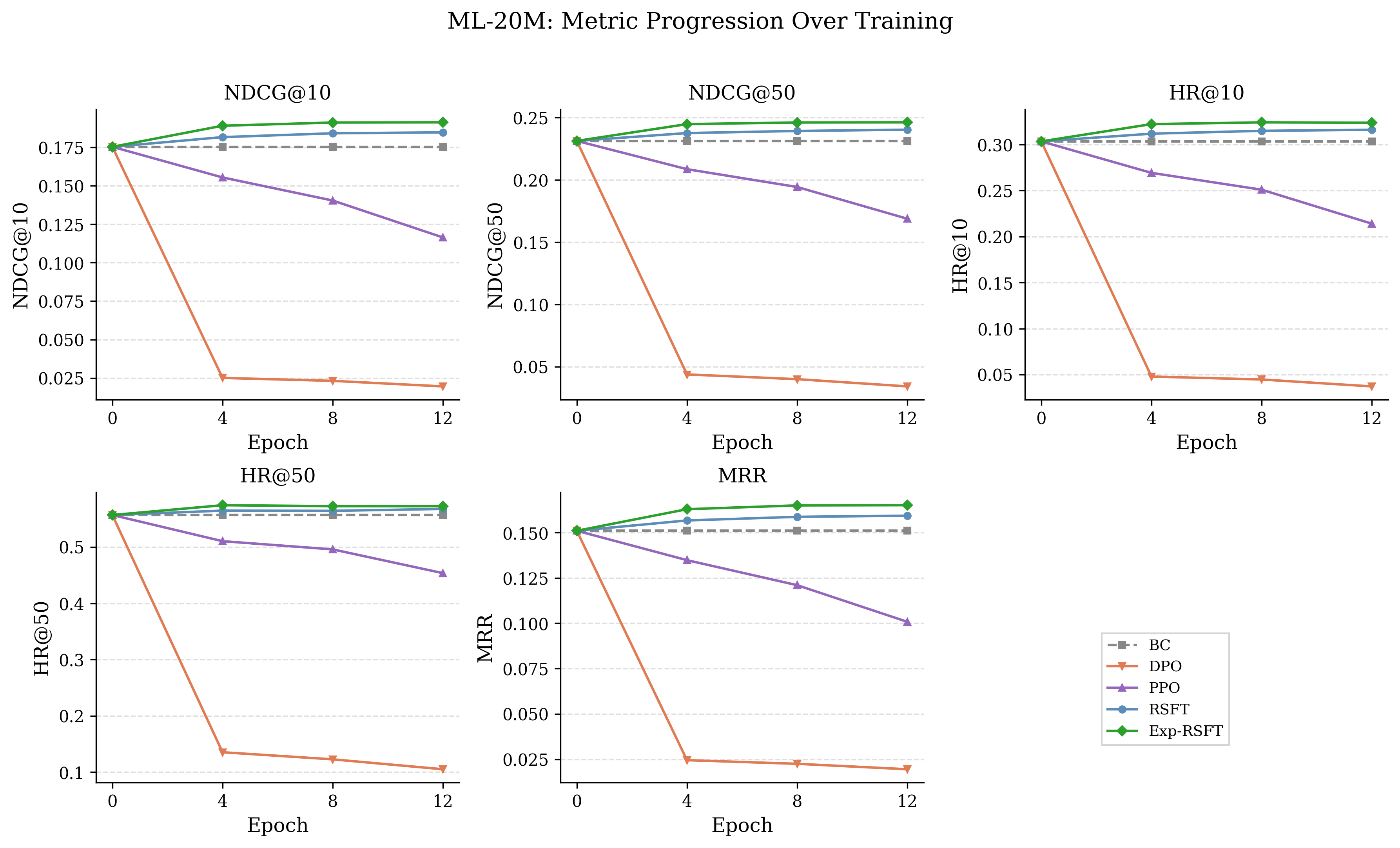}
  \caption{All metrics across training epochs for all algorithms on ML-20M.}
  \label{fig:metrics_all_ml-20m}
\end{figure*}

\begin{figure*}[t]
  \centering
  \includegraphics[width=0.72\textwidth]{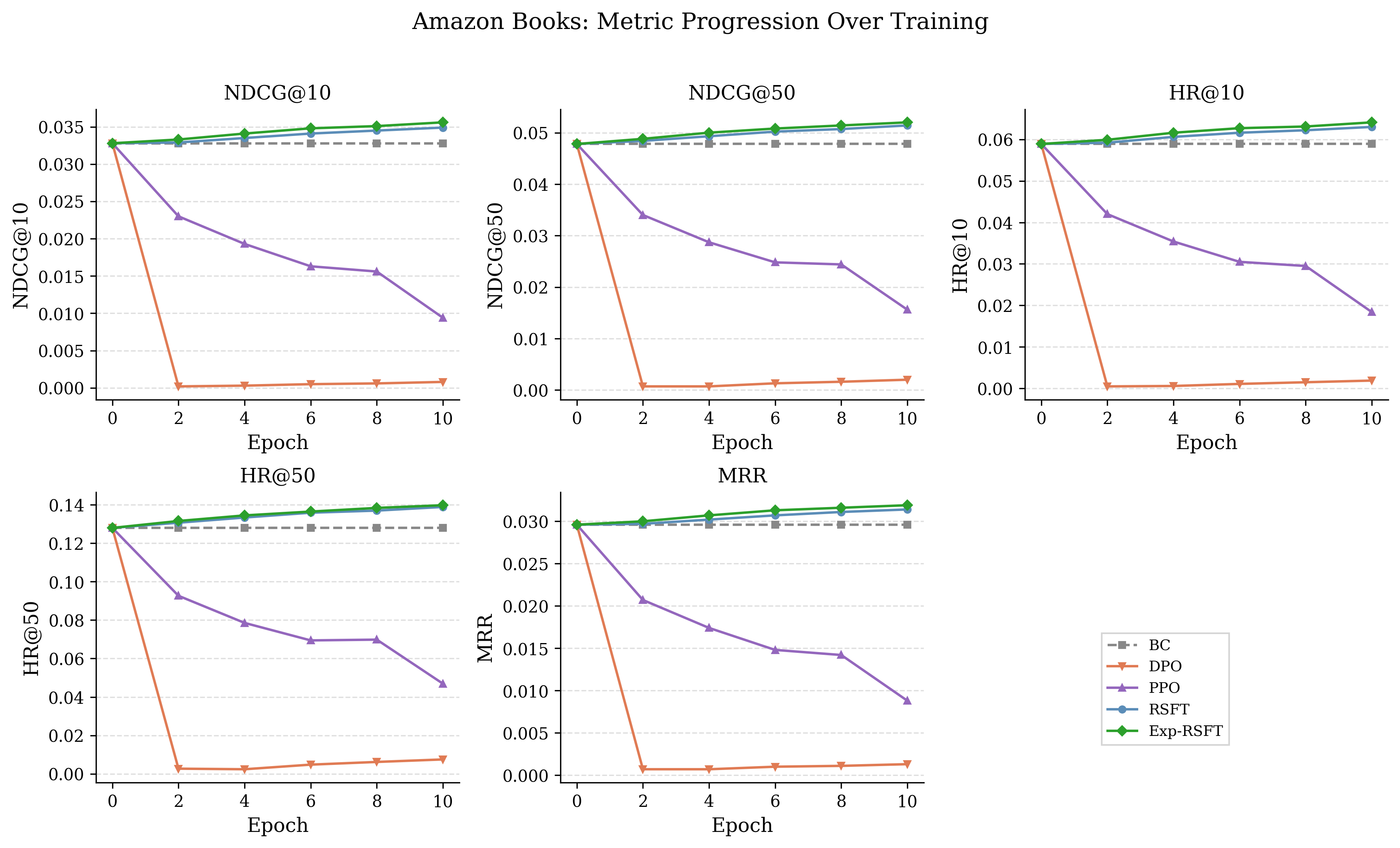}
  \caption{All metrics across training epochs for all algorithms on Amazon Books.}
  \label{fig:metrics_all_amazon-books}
\end{figure*}

\begin{figure*}[t]
  \centering
  \includegraphics[width=0.72\textwidth]{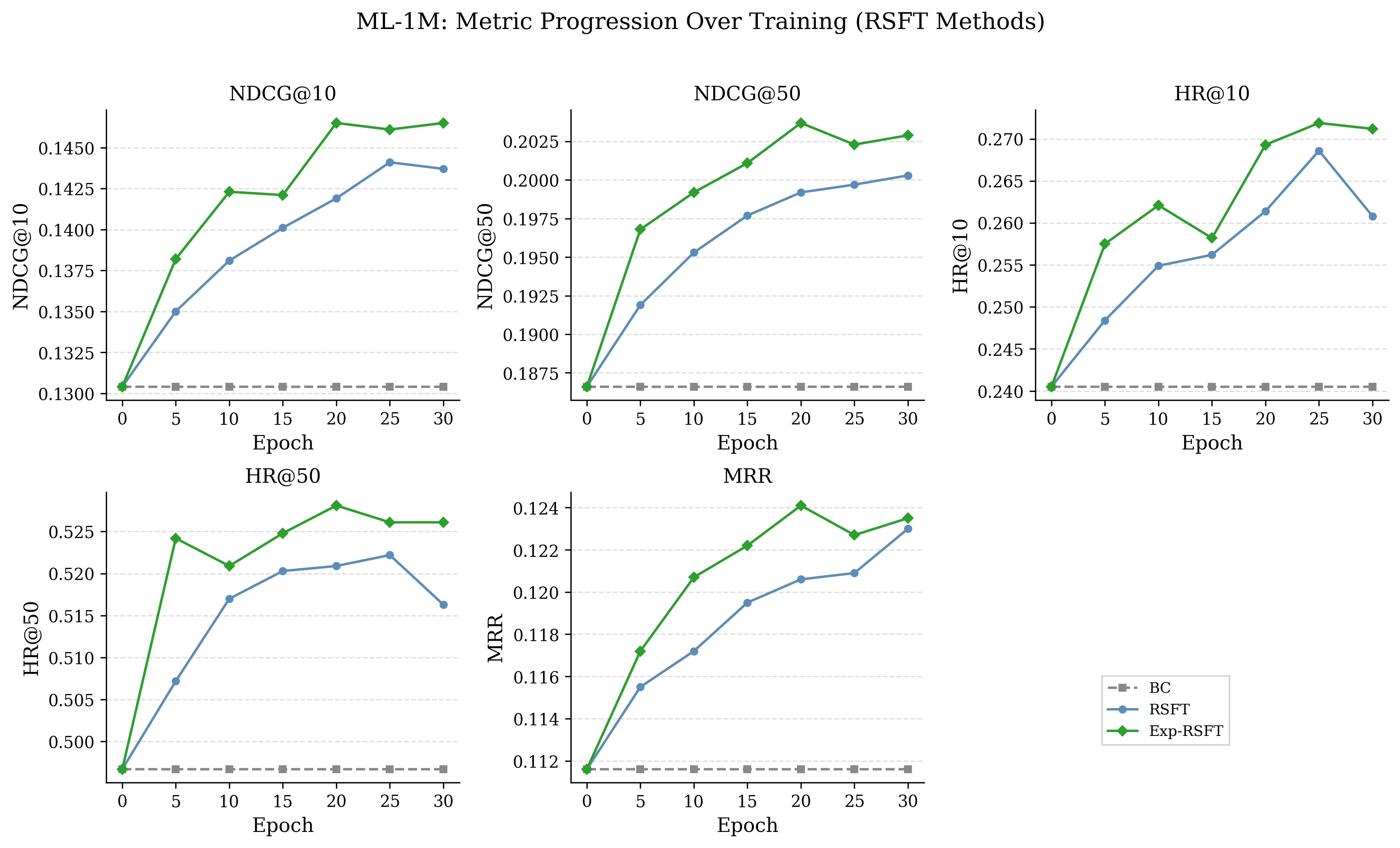}
  \caption{RSFT versus Exp-RSFT across training epochs on ML-1M.}
  \label{fig:metrics_zoom_ml-1m}
\end{figure*}

\begin{figure*}[t]
  \centering
  \includegraphics[width=0.72\textwidth]{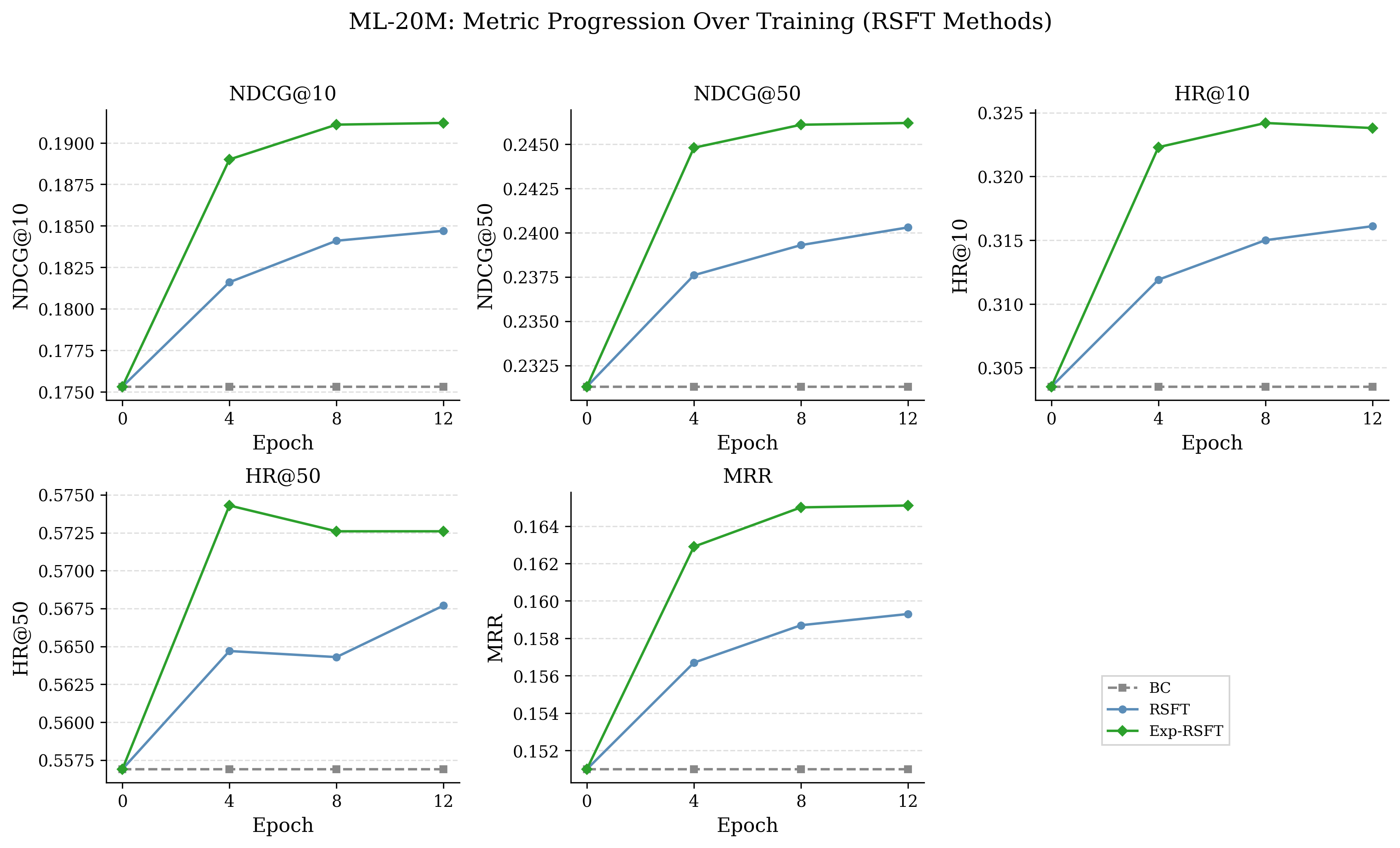}
  \caption{RSFT versus Exp-RSFT across training epochs on ML-20M.}
  \label{fig:metrics_zoom_ml-20m}
\end{figure*}

\begin{figure*}[t]
  \centering
  \includegraphics[width=0.72\textwidth]{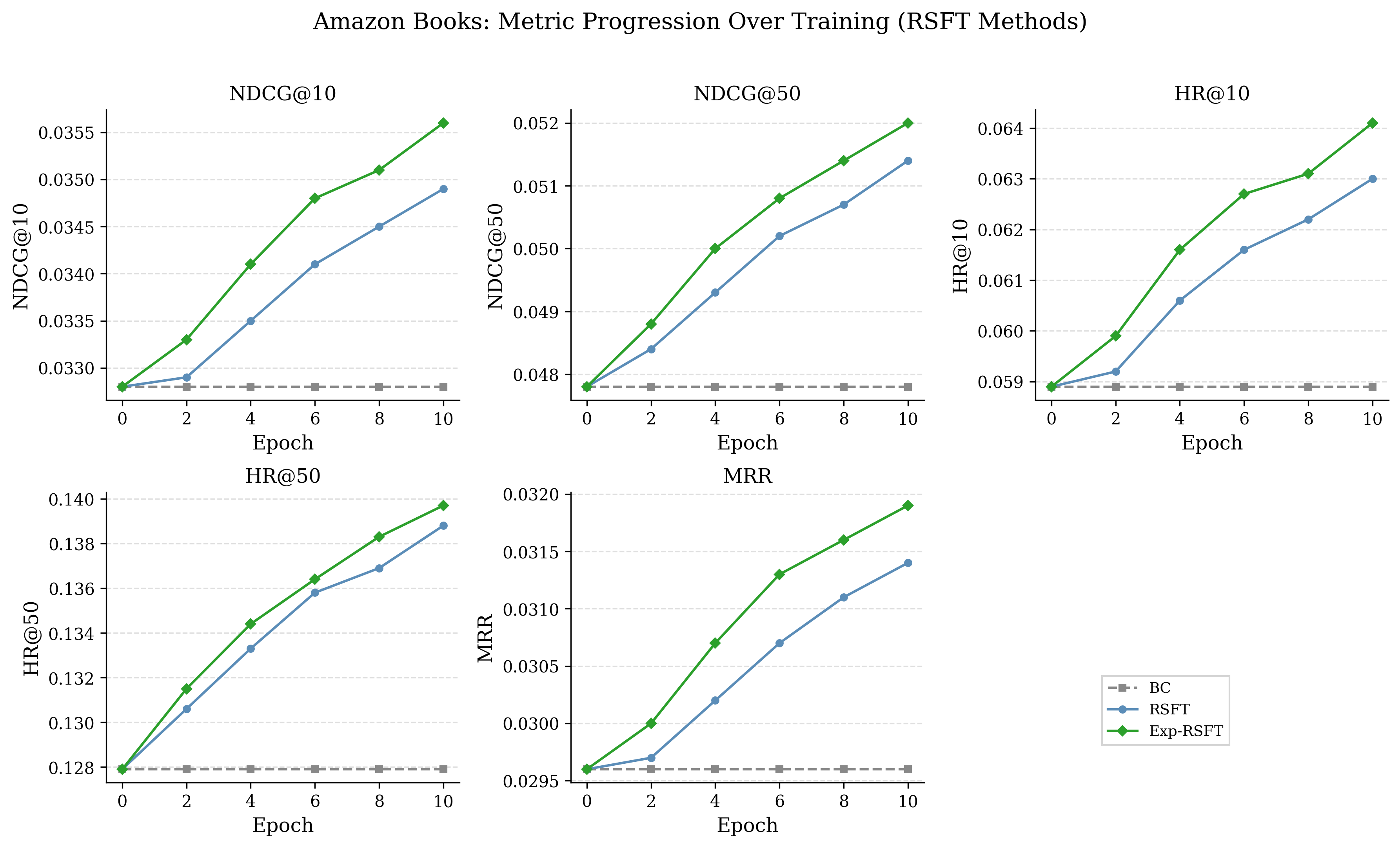}
  \caption{RSFT versus Exp-RSFT across training epochs on Amazon Books.}
  \label{fig:metrics_zoom_amazon-books}
\end{figure*}

\end{document}